\documentclass[runningheads,envcountsame]{llncs}
\usepackage[latin2]{inputenc}
\usepackage[english]{babel}
\usepackage{amsmath}
\usepackage{amssymb}
\usepackage{graphicx}
\usepackage{tabularx}
\usepackage{comment}

\usepackage{eucal}
\usepackage{mathrsfs}
\usepackage{amsfonts}
\newcommand{\prob}[1]{\textsc{LS-CSP$(#1)$}}

\newcommand{\ov}{\overline}
\def\zd{,\ldots,}
\newcommand{\proj}{\textup{pr}}
\newcommand{\dist}{\textup{dist}}
\newcommand{\CSP}[1]{\mbox{\rm CSP$(#1)$}}

\newcommand{\PTIME}{\mbox{\rm PTIME}}
\newcommand{\NP}{\mbox{\rm NP}}

\begin{document}

\title{On the hardness of losing weight}
\titlerunning{On the hardness of losing weight}

\author{Andrei Krokhin\inst{1} \and D\'{a}niel Marx\inst{2}}

\authorrunning{A.~Krokhin and D.~Marx}

\institute{Department of Computer Science, Durham University, Durham, DH1 3LE, UK
\email{andrei.krokhin@durham.ac.uk}\\
\and Department of Computer Science and Information Theory, Budapest University of
Technology and Economics, H-1521 Budapest, Hungary\\ \email{dmarx@cs.bme.hu}}

\maketitle

\begin{abstract}
  We study the complexity of local search for the Boolean constraint
  satisfaction problem (CSP), in the following form: given a CSP
  instance, that is, a collection of constraints, and a solution to
  it, the question is whether there is a better (lighter, i.e., having
  strictly less Hamming weight) solution
  within a given distance from the initial solution. We classify the
  complexity, both classical and parameterized, of such problems by a
  Schaefer-style dichotomy result, that is, with a restricted set of
  allowed types of constraints. Our results show that there is a
  considerable amount of such problems that are NP-hard, but
  fixed-parameter tractable when parameterized by the distance.
\end{abstract}

\section{Introduction}
Local search is one of the most widely used approaches to solving hard
optimization problems. The basic idea of local search is that one
tries to iteratively improve a current solution by searching for
better solutions in its ($k$-)neighborhood (i.e., within distance $k$
from it). Any optimization algorithm can be followed by a local search
phase, thus the problem of finding a better solution locally is of great
practical interest.  However, a brute force search of a
$k$-neighborhood is not feasible for large $k$, thus it is natural to study
the complexity of searching the $k$-neighborhood.

The constraint satisfaction problem (CSP) provides a framework in which it is
possible to express, in a natural way, many combinatorial problems encountered in
artificial intelligence and computer science. A constraint satisfaction problem is
represented by a set of variables, a domain of values for each variable, and a set
of constraints between variables. The basic aim in a constraint satisfaction problem
is then to find an assignment of values to the variables that satisfies the
constraints. Boolean CSP (when all variables have domain $\{0,1\}$), aka generalized
satisfiability, is a natural generalization of {\sc Sat} where constraints are given
by arbitrary relations, not necessarily by clauses. Local search methods for {\sc
Sat} and CSP are very extensively studied (see,
e.g.,~\cite{Dantsin02:local,Gu00:book,Hirsch00:local,Hoos06:handbook}).

Complexity classifications for various versions of CSP (and, in particular, for
versions of Boolean CSP) have recently attracted massive attention from researchers,
and one of the most popular directions here is to characterise restrictions on the
type of constraints that lead to problems with lower complexity in comparison with
the general case (see~\cite{Cohen06:handbook,Creignou01:book}). Such classifications
are sometimes called Schaefer-style because the first classification of this type
was obtained by T.J.~Schaefer in his seminal work~\cite{Schaefer78:complexity}. A
local-search related Schaefer-style classification for Boolean {\sc Max CSP} was
obtained in~\cite{Chapdelaine05:local}, in the context of local search complexity
classes such as PLS.

The hardness of searching the $k$-neighborhood (for any optimisation problem) can be
studied very naturally in the framework of parameterized
complexity~\cite{Downey99:book,Flum06:book}, as suggested
in~\cite{Fellows01:frontiers}; such a study for the traveling salesman problem (TSP)
was recently performed in~\cite{Marx08:TSP}. Parameterized complexity studies
hardness in finer detail than classical complexity. Consider, for example, two
standard NP-complete problems {\sc Minimum Vertex Cover} and {\sc Maximum Clique}. Both have
the natural parameter $k$: the size of the required vertex cover/clique. Both
problems can be solved in $n^{O(k)}$ time by complete enumeration. Notice that the
degree of the polynomial grows with $k$, so the algorithm becomes useless for large
graphs, even if $k$ is as small as 10. However, {\sc Minimum Vertex Cover} can be solved
in time $O(2^k\cdot n^2)$~\cite{Downey99:book,Flum06:book}. In other words, for
every fixed cover size there is a polynomial-time (in this case, quadratic)
algorithm solving the problem where the degree of the polynomial is independent of
the parameter. Problems with this property are called fixed-parameter tractable. The
notion of W[1]-hardness in parameterized complexity is analogous to NP-completeness
in classical complexity. Problems that are shown to be W[1]-hard, such as {\sc Maximum
Clique}~\cite{Downey99:book,Flum06:book}, are very unlikely to be fixed-parameter
tractable. A Schaefer-style classification of the basic Boolean CSP with respect to
parameterized complexity (where the parameter is the required Hamming weight of the solution)
was obtained in~\cite{Marx05:CSP}.

In this paper, we give a Schaefer-style complexity classification for
the following problem: given a collection of Boolean constraints, and
a solution to it, the question is whether there is a better (i.e.,
with smaller Hamming weight) solution within a given (Hamming)
distance $k$ from the initial solution. We obtain classification
results both for classical (Theorem~\ref{thm:poly}) and for
parameterized complexity (Theorem~\ref{thm:fpt}).  However, we would
like to point out that it makes much more sense to study this problem
in the parameterized setting. Intuitively, if we are able to decide in
polynomial time whether there is a better solution within distance
$k$, then this seems to be almost as powerful as finding the best
solution (although there are technicalities such as whether there is a
feasible solution at all). Our classification confirms this intuition:
searching the $k$-neighborhood is polynomial-time solvable only in
cases where finding the optimum is also polynomial-time solvable. On
the other hand, there are cases (for example, Horn constraints or
\textsc{1-in-3 Sat}) where the problem of finding the optimum is NP-hard,
but searching the $k$-neighborhood is fixed-parameter tractable.  This
suggests evidence that parameterized complexity is the right setting
for studying local search.

The paper is organized as follows. Section~\ref{sec:preliminaries}
reviews basic notions of parameterized complexity and Boolean CSP.
Section~\ref{sect:fpt} presents the classificiation with respect to
fixed-parameter tractability, while Section~\ref{sect:poly} deals with
polynomial-time solvability. The proofs omitted from
Section~\ref{sect:poly} can be found in the appendix.

\section{Preliminaries}\label{sec:preliminaries}

\textbf{Boolean CSP.} A {\em formula} $\phi$ is a pair $(V,C)$
consisting of a set $V$ of {\em variables} and a set $C$ of {\em
  constraints.} Each constraint $c_i\in C$ is a pair $\langle
\overline s_i, R_i
\rangle$, where $\overline s_i=( x_{i,1},\dots,  x_{i,r_i})$ is an $r_i$-tuple of
variables (the {\em constraint scope}) and $R_i\subseteq \{0,1\}^{r_i}$ is an $r_i$-ary Boolean
relation (the {\em constraint relation}). A function $f: V\to \{0,1\}$ is a {\em satisfying
  assignment} of $\phi$ if  $(f(x_{i,1}),\dots, f(x_{i,r_i}))$ is
in $R_i$ for every $c_i \in C$. Let $\Gamma$ be a set of Boolean
relations. A formula is a {\em $\Gamma$-formula} if every constraint
relation $R_i$ is in $\Gamma$. In this paper, $\Gamma$ is always a
finite set.
The {\em (Hamming) weight} $w(f)$ of
assignment $f$
is the number of variables $x$ with $f(x)=1$. The {\em distance} $\dist(f_1,f_2)$ of
assignments $f_1,f_2$ is the number of variables $x$ with
 $f_1(x)\neq f_2(x)$.


We recall various standard definitions concerning Boolean constraints (cf.~\cite{Creignou01:book}):
\begin{itemize}
\item $R$ is {\em 0-valid} if $(0,\dots,0)\in R$.
\item $R$ is {\em 1-valid} if $(1,\dots,1)\in R$.
\item $R$ is {\em Horn} or {\em weakly negative} if it can be
  expressed as a conjunction of clauses such that each clause
  contains at most one positive literal. It is known that $R$ is Horn
  if and only if it is {\em min-closed:} if
   $(a_1,\dots,a_r)\in R$ and $(b_1,\dots,b_r)\in R$, then
  $(\min(a_1,b_1),\dots, \min(a_r,b_r))\in R$.
\item $R$ is {\em affine} if it can be expressed as a conjunction
  of constraints of the form $x_1 + x_2 +\dots +x_t=b$, where $b\in
  \{0,1\}$ and addition is modulo 2. The number of tuples in an affine
  relation is always an integer power of 2.
\item $R$ is {\em width-2 affine} if it can be expressed as a conjunction
  of constraints of the form $x=y$ and $x\neq y$.
\item $R$ is {\em IHS-B$-$} (or {\em implicative hitting set bounded}) if
it can be represented by a conjunction of clauses of the form $(x)$, $(x \to y)$ and
$(\neg x_1 \vee \ldots \neg x_n)$, $n\ge 1$.
\item The relation $R_{\textsc{$p$-in-$q$}}$ (for $1 \le p \le q$) has arity $q$ and
  $R_{\textsc{$p$-in-$q$}}(x_1,\dots,x_q)$ is true if and only if exactly
  $p$ of the variables $x_1$, $\dots$, $x_q$ have value $1$.
\end{itemize}

The following definition is new in this paper. It plays a crucial role
in characterizing the fixed-parameter tractable cases for local search.
\begin{definition}
Let $R$ be a Boolean relation and $(a_1,\dots, a_r)\in R$. A set
$S\subseteq \{1,\dots,r\}$ is a {\em flip set} of $(a_1,\dots,a_r)$
(with respect to $R$) if $(b_1,\dots,b_r)\in R$ where $b_i=1-a_i$ for
$i\in S$ and $b_i=a_i$ for $i\not\in S$. We say that $R$ is {\em flip
  separable} if whenever some $(a_1,\dots,a_r)\in R$ has two flip sets
$S_1,S_2$ with $S_1\subset S_2$, then $S_2\setminus S_1$ is also a
flip set for $(a_1,\dots,a_r)$.
\end{definition}
It is easy to see that $R_{\textsc{1-in-3}}$ is flip separable: every flip set has
size exactly 2, hence $S_1\subset S_2$ is not possible.
Moreover, $R_{\textsc{$p$-in-$q$}}$ is also flip separable for every
$p\le q$.
Affine constraints are also flip separable: to see this,
it is sufficient to verify the definition only for the constraint $x_1+\dots
+x_r=0$.

The basic problem in CSP is to decide  if a formula has a satisfying assignment:
\begin{center}
\fbox{
\parbox{0.9\linewidth}{
\textsc{CSP($\Gamma$)}\\[5pt]
\begin{tabular}{rl}
{\em Input:} & A $\Gamma$-formula $\phi$.\\[3pt]
{\em Question:} & Does $\phi$ have a satisfying assignment?
\end{tabular}
}}
\end{center}

Schaefer completely characterized the complexity of \textsc{CSP($\Gamma$)} for every
finite set $\Gamma$ of Boolean relations \cite{Schaefer78:complexity}. In
particular, every such problem is either in \PTIME\ or \NP-complete, and there is a
very clear description of the boundary between the two cases.
%

Optimization versions of Boolean CSP were investigated in
\cite{Creignou01:book,Crescenzi02:Hamming}. A straightforward way to obtain an
optimization problem is to relax the requirement that every constraint is satisfied,
and ask for an assignment maximizing the number of satisfied constraints. Another
possibility is to ask for a solution with minimum/maximum weight. In this paper, we
investigate the problem of minimizing the weight. As we do not consider the
approximability the problem, we define here only the decision version:
\begin{center}
\fbox{
\parbox{0.9\linewidth}{
\textsc{Min-Ones($\Gamma$)}\\[5pt]
\begin{tabular}{rl}
{\em Input:} &A $\Gamma$-formula $\phi$ and an integer $W$.\\[3pt]
{\em Question:} &Does $\phi$ have a satisfying assignment $f$ with
$w(f)\le W$?
\end{tabular}
}}
\end{center}

The characterization of the approximability of finding a minimum
weight satisfying assignment for a $\Gamma$-formula can be found in
\cite{Creignou01:book}. Here we state only the classification of
polynomial-time solvable and NP-hard cases:
\begin{theorem}[\cite{Creignou01:book}]
Let $\Gamma$ be a finite set of Boolean relations. \textsc{Min-Ones($\Gamma$)} is
solvable in polynomial time if one the following holds, and \NP-complete otherwise:
\begin{itemize}
\item Every $R\in \Gamma$ is 0-valid.
\item Every $R\in \Gamma$ is Horn.
\item Every $R\in \Gamma$ is width-2 affine.
\end{itemize}
\end{theorem}

A Schaefer-style characterization of the approximability of finding two satisfying
assignments to a formula with a largest distance between them was obtained
in~\cite{Crescenzi02:Hamming}, motivated by the blocks world problem from KR, while
a Schaefer-style classification of the problem of deciding whether a given
satisfying assignment to a given CSP instance is component-wise minimal was
presented in~\cite{Kirousis03:minimal}, motivated by the circumscription formalism
from AI.

The main focus of the paper is the local search version of minimizing weight:
\begin{center}
\fbox{
\parbox{0.9\linewidth}{
\prob{\Gamma}\\[5pt]
\begin{tabular}{rp{0.85\linewidth}}
{\em Input:} & A $\Gamma$-formula $\phi$, a satisfying assignment $f$,
and an integer $k$.\\[3pt]
{\em Question:} & Does $\phi$ have a
  satisfying assignment $f'$ with $w(f')<w(f)$ and $\dist(f,f')\le k$?
\end{tabular}
}}
\end{center}

LS in the above problem stands for both ``local search'' and  ``lighter solution.''

Observe that the satisfying assignments of an $(x \vee y)$-formula correspond to the
vertex covers of the graph where the variables are the vertices and the edges are
the constraints. Thus $\prob{\{x\vee y\}}$ is the problem of reducing the size of a
(given) vertex cover by including and excluding a total of at most $k$ vertices. As
we shall see (Lemma~\ref{lem:orhard}), this problem is W[1]-hard, even for bipartite
graphs. This might be of independent interest.

\textbf{Parameterized complexity.}
In a {\em parmeterized problem,} each instance contains an integer $k$
called the {\em parameter}. A parameterized
 problem  is \textit{fixed-param\-eter tractable (FPT)} if can be
 solved by
 an algorithm with running time $f(k)\cdot n^c$, where $n$ is the
 length of the input, $f$ is an
 arbitrary (computable) function depending only on $k$, and
 $c$ is a constant independent of $k$.


A large fraction of NP-complete problems is known to be FPT. On the
other hand, analogously to NP-completeness in classical complexity, the theory of
W[1]-hardness can be used to give strong evidence that certain
problems are unlikely to be fixed-parameter tractable.
 We omit the somewhat technical
definition of the complexity class W[1], see \cite{Downey99:book,Flum06:book} for
details. Here it will be sufficient to know that there are many
problems, including \textsc{Maximum Clique}, that were proved to be
W[1]-hard.
To prove that a parameterized problem is W[1]-hard, we have to
present a parameterized reduction from a known W[1]-hard
problem. A {\em parameterized reduction} from problem $L_1$ to problem
$L_2$ is a function that transforms
a problem instance $x$ of $L_1$ with parameter $k$ into a problem instance $x'$
of $L_2$ with parameter $k'$ in such a way that
\begin{itemize}
\item $x'$ is a yes-instance of $L_2$ if and only if $x$ is a
  yes-instance of $L_1$,
\item $k'$ can be bounded by a function of $k$, and
\item the transformation can be computed in time $f(k)\cdot |x|^c$ for
  some constant $c$ and function $f(k)$.
\end{itemize}
It is easy to see that if there is a parameterized reduction from
$L_1$ to $L_2$, and $L_2$ is FPT, then it
follows that $L_1$ is FPT as well.


\section{Characterizing fixed-parameter tractability}\label{sect:fpt}

In this section, we completely characterize those finite sets $\Gamma$ of
Boolean relations for which \prob{\Gamma} is
fixed-parameter tractable.

\begin{theorem}\label{thm:fpt}
Let $\Gamma$ be a finite set of Boolean relations. The problem \prob{\Gamma} is in
\textup{FPT} if every relation in $\Gamma$ is Horn or every relation in $\Gamma$ is flip
separable. In all other cases, \prob{\Gamma} is \textup{W[1]}-hard.
\end{theorem}

First we handle the fixed-parameter tractable cases
(Lemmas~\ref{horn-fpt} and~\ref{lem:flip-separ-FPT})
\begin{lemma}\label{horn-fpt}
If every $R\in \Gamma$ is Horn, then
\prob{\Gamma} is \textup{FPT}.
\end{lemma}
\begin{proof}
Observe that if there is a solution $f'$, then we can assume that
$f'(x)\le f(x)$ for every variable $x$: by defining $f''(x):=\min
\{f(x),f'(x)\}$, we get that $f''$ is also satisfying (as every
$R\in \Gamma$ is min-closed) and $\dist(f'',f)\le \dist(f',f)$.
Thus we can restrict our search to solutions that can be obtained
from $f$ by changing some 1's to 0's, but every 0 remains
unchanged.

Since $w(f')<w(f)$, there is a variable $x$ with $f(x)=1$ and $f'(x)=0$.  For
every variable $x$ with $f(x)=1$, we try to find a solution $f'$ with
$f'(x)=0$ using a simple bounded-height search tree algorithm. For a
particular $x$, we proceed as follows. We start with initial
assignment $f$. Change the value of $x$ to $0$. If
there is a constraint $\langle (x_1,\dots,x_r),R\rangle$ that is not
satisfied by the new assignment, then we select one of the variables
$x_1$, $\dots$, $x_r$ that has value 1, and change it to 0. Thus at
this point we branch into at most $r$ directions. If the assignment is
still not satisfying, the we branch again on the variables of some
unsatisfied constraint. The branching factor of the resulting search
tree is at most $r_\text{max}$, where $r_\text{max}$ is the maximum
arity of the relations in $\Gamma$. By the observation above, if there
is a solution, then we find a solution on the first $k$ levels of the
search tree. Therefore, we can stop the search on the $k$-th level,
implying that we visit at most $r_\text{max}^{k+1}$ nodes of the search
tree. The work to be done at each node is polynomial in the size $n$
of the input, hence the total running time is $r_\text{max}^{k+1} \cdot
n^{O(1)}$.
\qed\end{proof}
If every $R\in \Gamma$ is not only Horn, but IHS-B$-$ (which is a
subset of Horn), then the algorithm of Lemma~\ref{horn-fpt} actually runs in
polynomial time:
\begin{corollary}\label{lem:cor-ihbs}
If every $R\in \Gamma$ is IHS-B$-$, then
\prob{\Gamma} is in \PTIME.
\end{corollary}
\begin{proof}
  We can assume that every constraint is either $(x)$, $(x\to y)$, or
  $(\bar x_1 \vee \dots \vee \bar x_r)$. If a constraint $(\bar x_1
  \vee \dots \vee \bar x_r)$ is satisfied in the initial assignment
  $f$, then it remains satisfied after changing some 1's to 0.
  Observe that if a constraint $(x)$ or $(x\to y$) is not satisfied,
  then at most one variable has the value 1. Thus there is no
  branching involved in the algorithm of Lemma~\ref{horn-fpt}, making
  it a polynomial-time algorithm.  \qed\end{proof}

For flip separable relations, we give a very similar branching
algorithm. However, in this case the correctness of the algorithm
requires a nontrivial argument.
\begin{lemma}\label{lem:flip-separ-FPT}
If every $R\in \Gamma$ is flip separable, then
\prob{\Gamma} is FPT.
\end{lemma}
\begin{proof}
Let $(\phi, f,k)$ be an instance of $\prob{\Gamma}$.
  If $w(f')<w(f)$ for some assignment $f'$, there is a variable $x$ with $f(x)=1$ and
  $f'(x)=0$.  For every variable $x$ with $f(x)=1$, we try to find a
  solution $f'$ with $f'(x)=0$ using a simple bounded-height search
  tree algorithm. For each such $x$, we proceed as follows. We
  start with the initial assignment $f$ and set the value of $x$ to $0$.
  Iteratively do the following: (a) if there is a constraint in $\phi$ 
  that is not satisfied by the current assignment and such that the value of
some variable in it has not been flipped yet (on this branch),
then we select one of such
  variables, and flip its
  value; (b) if there is no such constraint, but the current assignment is not satisfying then
  we move to the next branch; (c) if
  every constraint is satisfied, then either we found a required solution or else
 we move to the next branch. If a required solution is not found
on the first $k$ levels of the search tree then the algorithm
reports that there is no required solution.

  Assume that $(\phi,f,k)$ is a yes-instance.
  We claim that if $f'$ is a required solution with minimal distance from $f$,
  then some branch of the algorithm finds it.
  Let $X$ be the set of variables on which $f$ and $f'$ differ, so $|X|\le k$.
  We now show that on the first $k$ levels of the
  search tree, the algorithm finds some satisfying assignment $f_0$
  (possibly heavier than $f$) that
  differs from $f$ only on a subset $X_0\subseteq X$ of variables. To
  see this, assume that at some node of the search tree, the current
  assignment differs from the initial assignment only on a subset of
  $X$; we show that this remains true for at least one child of the
  node. If we branch on the variables $(x_1,\dots,x_r)$ of an
  unsatisfied constraint, then at least one of its variables, say
  $x_i$, has a value different from $f'$ (as $f'$ is a satisfying
  assignment). It follows that $x_i\in X$: otherwise the current value
  of $x_i$ is $f(x_i)$ (since so far we changed variables only in $X$)
  and $f(x_i)=f'(x_i)$ (by the definition of $X$), contradicting the
  fact that current value of $x_i$ is different from $f(x_i)$. Thus if
  we change  variable $x_i$, it remains true that only variables
  from $X$ are changed.  Since $|X|\le k$, this branch of the
  algorithm has to find some satisfying assignment $f_0$.

  If $w(f_0)<w(f)$, then, by the choice of $f'$, we must have $f_0=f'$.
  Otherwise, let $X_0\subseteq X$ be
  the set of variables where $f$ and $f_0$ differ and let $f''$ be the
  assignment that differs from $f$ exactly on the variables
  $X\setminus X_0$. From the fact that every constraint is flip
  separable, it follows that $f''$ is a satisfying assignment. We
  claim that $w(f'')<w(f)$. Indeed, if changing the values of the
  variables in $X$ decreases the weight and changing the values in
  $X_0$ does not decrease the weight, then the set $X\setminus X_0$ has to
  decrease the weight. This contradicts the assumption that $f'$ is a
  solution whose distance from $f$ is minimal: $f''$ is a solution
  with distance $|X\setminus X_0|<|X|$.  Thus it is sufficient to
  investigate only the first $k$ levels of the search tree. As in the proof of
  Lemma~\ref{horn-fpt}, the branching factor of the tree is at
  most $r_\text{max}$, and the algorithm runs in time
   $r_\text{max}^{k+1} \cdot n^{O(1)}$.
  \qed\end{proof}

All the hardness proofs in this section are based on the following
lemma:
\begin{lemma}\label{lem:orhard}
\prob{\{x\vee y\}} is \textup{W[1]}-hard.
\end{lemma}
\begin{proof}
The proof is by reduction from a variant of \textsc{Maximum Clique}: given
a graph $G(V,E)$ with a distinguished vertex $x$ and an integer $t$, we
have to decide whether $G$ has a clique of size $t$ that contains
$x$. It is easy to see that this problem is W[1]-hard.
Furthermore, it can be assumed that $t$
is odd. Let $n$ be the number of vertices of $G$ and let $m$ be the
number of edges. We construct a formula $\phi$ on $m+n(t-1)/2-1$
variables and a satisfying assignment $f$ such that $G$ has a clique of size $t$
containing $x$ if
and only if $\phi$ has a satisfying assignment $f'$ with $w(f')<w(f)$
and distance at most $k:=t(t-1)-1$ from $f$.

Let $d=(t-1)/2$ (note that $t$ is odd). The formula $\phi$ has $d$ variables $v_1$, $\dots$, $v_d$
for each vertex $v\neq x$ of $G$ and a variable $u_e$ for each edge $e$
of $G$. The distinguished vertex $x$ has only $d-1$ variables $x_1$,
$\dots$, $x_{d-1}$.  If a vertex $v$ is the endpoint of an edge $e$,
then for every $1 \le i \le d$ (or $1 \le i \le d-1$, if $v=x$), we add the constraint $u_e\vee
v_{i}$. Thus each variable $u_e$
is in $2d-1$ or $2d$ constraints
(depending on whether $x$ is the endpoint of $e$ or not).
Set $f(u_e)=1$ for every $e\in E$ and $f(v_{i})=0$ for every $v\in V$,
$1\le i \le d$. Clearly, $f$ is a satisfying assignment.

Assume that $G$ has a clique $K$ of size $t$ that includes
$x$. Set $f'(v_i)=1$ for every $v\in K$ ($1 \le i \le d$) and set
$f'(u_e)=0$ for every edge $e$ in $K$; let $f'$ be the same
as $f$ on every other variable. Observe that $f'$ is also a satisfying
assignment: if a variable $u_e$ was changed to 0 and there is a
constraint $u_e \vee v_i$, then $v\in K$ and hence
$f'(v_i)=1$. We have $w(f')<w(f)$:
$dt-1$ variables were changed to $1$ (note that $x\in K$) and
$t(t-1)/2=dk$ variables were changed to $0$. Moreover, the distance of
$f$ and $f'$ is exactly $dt-1+t(t-1)/2=t(t-1)-1=k$.

Assume now that $f'$ satisfies the requirements.  Let $K$ be the set
of those vertices $v$ in $G$ for which $f'(v_i)=1$ for every $i$. We
claim that $K$ is a clique of size $t$ in $G$.  Observe that there are
at least $d|K|-1$ variables $v_i$ with $f'(v_i)>f(v_i)$ and
$f'(u_e)<f(u_e)$ is possible only if both endpoints of $e$ are in $K$,
i.e., $e$ is in the set $E(K)$ of edges in $K$. Thus
$w(f')<w(f)$ implies $d|K|-1<|E(K)|\le |K|(|K|-1)/2$, which is only
possible if $|K|\ge t$.
 If
$|K|>t$, then $f'(v_i)>f(v_i)$ for at least $(t+1)d-1$ variables,
hence there must be at least that many variables $u_e$ with
$f'(u_e)<f(u_e)$. Thus the distance of $f$ and $f'$ is at least
$2(t+1)d-2>t(t-1)-1$.
Therefore, we can assume $|K|=t$.  Now $dt-1 < |E(K)| \le
|K|(|K|-1)/2=t(t-1)/2$ is only possible if $|E(K)|=t(t-1)/2$ (i.e., $K$
is a clique) and it follows that there are exactly $dt-1$ variables $v_i$ with
$f'(v_i)>f(v_i)$ (i.e., $x\in K$).
\qed\end{proof}
Now we are ready to present the main hardness proof of the section:
\begin{lemma}\label{lem:W1-main}
If $\Gamma$ contains a relation $R_1$ that is not Horn and a
relation $R_2$ that is not flip separable, then \prob{\Gamma} is \textup{W[1]}-hard.
\end{lemma}
\begin{proof}
  The proof is by reduction from \prob{\{x\vee y\}}. Let
  $(phi_1,f_1,k)$ be an instance of \prob{\{x\vee y\}}, i.e., every
  constraint relation in formula $\phi_1=(V,C)$ is $(x \vee y)$.
  Since $R_1$ is not min-closed, we can assume (by
  permuting the variables) that for some $r_1,r_2\ge 1$, $r_3,r_4\ge
  0$, if we define
\[
R'_1(x,y,w_0,w_1)=R_1(\overbrace{x,\dots,x}^{r_1},\overbrace{y,\dots,y}^{r_2},\overbrace{w_0,\dots,w_0}^{r_3},\overbrace{w_1,\dots,w_1}^{r_4}),
\]
then
$(0,1,0,1),(1,0,0,1)\in R'_1$, but $(0,0,0,1)\not\in R'_1$.
Since $R'_1$ is obtained from $R_1$ by identifying variables, we can
use the relation $R'_1$ when specifying instances of \prob{\Gamma}.
We consider two cases:

\textbf{Case 1:} $(1,1,0,1)\in R'_1$. In this case $R'_1(x,y,0,1)=x
\vee y$, hence it is easy to simulate \prob{\{x \vee y\}}.  The only
difficulty is how to simulate the constants 0 and 1. We do this as
follows. Let us construct a formula $\phi_2$ that has every variable
of $V$ and new variables $q^j_0$, $q^j_1$ for every $1 \le j \le k+1$
(these new variables will play the role of the constants). We define
assignment $f_2$ of $\phi_2$ by setting $f_2(x)=f_1(x)$ for $x\in V$
and $f_2(q^j_0)=0$ and $f_2(q^j_1)=1$ for $1 \le j \le k+1$.  For $1
\le a,b,c \le k+1$, we add constraint
$c^1_{a,b,c}=R'_1(q^{a}_1,q^b_0,q^b_0,q^c_1)$, it is clearly satisfied
by assignment $f_2$.  To simulate a constraint $x\vee y$, we add
$c^2_{x,y,j}=R'_1(x,y,q^j_0,q^1_1)$ for every $1 \le j \le k+1$.

It is easy to see that if there is a solution $f'_1$ for the original
instance $(\phi_1,f_1,k)$, then by setting $f'_2(x)=f'_1(x)$ for every
$x\in V$ and $f'_2(q^j_0)=0$, $f'_2(q^j_1)=1$ for every $1 \le j \le
k+1$ gives a solution $f'_2$ for the constructed instance
$(\phi_2,f_2,k)$. We claim the converse is also true: if $f'_2$ is a
solution for the instance $(\phi_2,f_2,k)$, then the restriction of $f'_2$ to
$V$ gives a solution for $(\phi_1,f_1,k)$. Since the distance of $f_2$ and
$f'_2$ is at most $k$, there are $1 \le b,c \le k+1$ with
$f'_2(q^b_0)=0$ and $f'_2(q^c_1)=1$.  Because of the constraint
$c^1_{a,b,c}$, we have that $f'_2(q^a_1)=1$ for every $1\le a \le
k+1$. It follows that $f'_2$ restricted $V$ is a satisfying assignment
of $\phi_1$: for every constraint $x\vee y \in C$, the constraint
$c^2_{x,y,b}$ prevents the possibility $f'_2(x)=f'_2(y)=0$. We have
seen that $f'_2(q^j_0)\ge f_2(q^j_0)$ and $f'_2(q^j_1)\ge f_2(q^j_1)$
for every $1 \le j \le k+1$. Now $w(f'_2)<w(f_2)$ implies that the
weight of $f'_2$ on $V$ has to be less than the weight of $f_2$ on
$V$. Thus $w(f'_1)<w(f_1)$.

\textbf{Case 2:} $(1,1,0,1)\not \in R'_1$, which means that
$R'_1(x,y,0,1)$ is $x\neq y$. In this case we have to rely on the fact
that $R_2$ is not flip separable to simulate the constraint $x\vee y$.
We construct formula $\phi_2$ and its satisfying assignment $f_2$ as
follows. Each variable $x$ is replaced by 3 variables $x_1$, $x_2$,
$x_3$.  We set $f_2(x_1)=f_2(x_2)=f_1(x)$ and $f_2(x_3)=1-f_1(x)$.
Furthermore, for $1\le j \le 3k+1$, we add the variables $q^j_0$ and
$q^j_1$ and set $f_2(q^j_0)=0$ and $f_2(q^j_1)=1$.

For every $1 \le a,b,c \le 3k+1$, we add the constraint
$c^1_{a,b,c}=R'_1(q^{a}_1,q^b_0,q^b_0,q^c_1) $, as in the
previous case. For every $x\in V$, $1 \le j \le 3k+1$, and $\ell=1,2$,
we add $c^2_{x,\ell,j}=R'_1(x_\ell,x_3,q^j_0,q^1_1)$, as we shall see,
the role of these constraints is to ensure $f'_2(x_1)=f'_2(x_2)\neq
f'_2(x_3)$.

Since $R_2$ is not flip separable, there is a tuple
$(s_1,\dots,s_r)\in R_2$ that has flip sets $S_1\subset S_2$, but
$S_2\setminus S_1$ is not a flip set. For every constraint $x\vee y$
of $\phi_1$, we add $3k+1$ constraints to $\phi_2$ as follows. First,
for $1 \le i \le r$ and $1 \le j \le 3k+1$, we define variable $v^j_i$
as
\[
v^j_i = \begin{cases}
x_1 & \text{if $i\in S_1$ and $s_i=0$},\\
x_3 & \text{if $i\in S_1$ and $s_i=1$},\\
y_1 & \text{if $i\in S_2\setminus S_1$ and $s_i=1$},\\
y_3 & \text{if $i\in S_2\setminus S_1$ and $s_i=0$},\\
q^1_1 & \text{if $i\not\in S_2$ and $s_i=1$},\\
q^j_0 & \text{if $i\not\in S_2$ and $s_i=0$}.
\end{cases}
\]
For every $1 \le j \le 3k+1$, we add the constraint
$c^3_{x,y,j}=R_2(v^j_1,\dots,v^j_{r}).$
For example, assume that $(0,1,0,1)\in R_2$ and this tuple has
flip sets $S_1=\{1,2\}$ and $S_2=\{1,2,3,4\}$, but $S_2\setminus S_1=\{3,4\}$
is not a flip set. This means that $(0,1,0,1),(1,0,1,0),(1,0,0,1)\in
R_2$ and $(0,1,1,0)\not\in R_2$. In this case,
constraint $c^3_{x,y,j}$ is $R_2(x_1,x_3,y_3,y_1)$. Assuming
$f(x_1)\neq f(x_3)$ and $f(y_1)\neq f(y_3)$, any combination of values
on $x_1$ and $y_1$ satisfies the constraint, except if
$f(x_1)=f(y_1)=0$. Thus the constraint effectively acts as a
constraint $x_1\vee y_1$.

Finally, we set the maximum allowed distance to $k':=3k$. This
completes the description of the constructed instance $(\phi_2,f_2,k')$.

Assume first that $f'_1$ is a solution for the instance $(\phi_1,f_1,k)$. Define
$f'_2(x_1)=f'_2(x_2)=f'_1(x)$ and $f'_2(x_3)=1-f'_1(x)$ for every
$x\in V$, and define $f'_2(q^j_0)=0$, $f'_2(q^j_1)=1$ for every $1\le
j\le 3k+1$. The fact $w(f'_1)<w(f_1)$ implies $w(f'_2)<w(f_2)$.
Furthermore, the distance of $f_2$ and $f'_2$ is exactly three times
the distance of $f_1$ and $f'_1$, i.e., at most $3k$.  We claim that
$f'_2$ satisfies the constraints of $\phi_2$. This is easy to see for
$c^1_{a,b,c}$ and $c^2_{x,\ell,j}$. For 
$c^3_{x,y,j}$, this can be seen as follows:
\begin{itemize}
\item
If $f'_2(x)=0$, $f'_2(y)=1$, then because $(s_1,\dots, s_r)\in
  R_2$.
\item If $f'_2(x)=1$, $f'_2(y)=0$, then because $S_2$ is a flip set.
\item If $f'_2(x)=1$, $f'_2(y)=1$, then because $S_1$ is a flip set.
\end{itemize}

For the other direction, assume that $f'_2$ is a solution for
instance $(\phi_2,f_2,k')$. Define $f'_1(x)=f'_2(x_1)$ for every $x\in V$; we claim that
$f'_1$ is a solution for instance $(\phi_1,f_1,k)$.  Since the distance of $f_2$ and
$f'_2$ is at most $3k$, there are $1 \le b,c \le 3k+1$ with
$f'_2(q^b_0)=0$ and $f'_2(q^c_1)=1$. Because of the constraint
$c^1_{a,b,c}$, we have that $f'_2(q^a_1)=1$ for every $1\le a \le
3k+1$. The constraints $c^2_{x,1,b}$ and $c^2_{x,2,b}$ ensure that
$f'_2(x_1)=f'_2(x_2)=1-f'_2(x_3)$ (since $(0,0,0,1)\not\in R'_1$ and
$(1,1,0,1)\not\in R'_1$). It follows that the distance of
$f_1$ and $f'_1$ is at most $k$: $f_1(x)\neq f'_1(x)$ implies
$f_2(x_\ell)\neq f'_2(x_\ell)$ for $\ell=1,2,3$, hence this can hold
for at most $k$ different $x\in V$. Moreover, $w(f'_1)<w(f_1)$: this
follows from the facts $w(f'_2)<w(f_2)$ and  $f'_2(q^j_0)\ge f_2(q^k_0)$,
$f'_2(q^j_1)\ge f_2(q^k_1)$  ($1\le j \le 3k+1$).

We claim that every constraint $x\vee y$ of $\phi_1$ is
satisfied. Assume that $f'_1(x)=f'_1(y)=f'_2(x_1)=f'_2(y_1)=0$. Now
 $c^3_{x,y,b}$ is not satisfied: this follows from the
fact that $S_2\setminus S_1$ is not a flip set for $(s_1,\dots, s_r)$
(with respect to $R_2$).
\qed\end{proof}


\section{Characterizing polynomial-time solvability}\label{sect:poly}

In this section, we completely characterize those finite sets $\Gamma$ of Boolean
relations for which \prob{\Gamma} is polynomial-time solvable.

\begin{theorem}\label{thm:poly}
Let $\Gamma$ be a finite set of Boolean relations. The problem \prob{\Gamma} is in
$\PTIME$ if every relation in $\Gamma$ is IHS-B$-$ or every relation in $\Gamma$ is
width-2 affine. In all other cases, \prob{\Gamma} is $\NP$-hard.
\end{theorem}
\begin{proof}
  If every relation in $\Gamma$ is IHS-B$-$, then
  Corollary~\ref{lem:cor-ihbs} gives a polynomial-time algorithm.
 If every relation in $\Gamma$ is width-2 affine then the following
 simple algorithm solves \prob{\Gamma}: for a given instance
   $(\phi,f,k)$, compute the graph whose vertices are the variables in $\phi$
   and two vertices are connected if there is a constraint in $\phi$
   imposed on them. If there is a connected component of this graph
   which has at most $k$ vertices and such that $f$ assigns more 1's
   in this component than it does 0's, then flipping the values in
   this component gives a required lighter solution. If such a
   component does not exists, then there is no lighter solution within
   distance $k$ from  $f$.

   By Lemma~\ref{lem:W1-main}, if $\Gamma$ contains a relation that is
   not Horn and a relation that is not flip separable then
   \prob{\Gamma} is $\NP$-hard.  (Note that the proof is actually a
   polynomial-time reduction from an \NP-hard problem.)  Therefore, we
   can assume that every relation in $\Gamma$ is Horn or every
   relation in $\Gamma$ is flip separable. We now give the proof for
   the former case, while the proof for the latter case can be found
   in Appendix.

Assume now that $\Gamma$ is Horn, and there is a relation $R\in \Gamma$ that is not
IHS-B$-$. We prove that $\prob{\{R\}}$ is NP-hard. It is shown in the proof of
Lemma~5.27 of~\cite{Creignou01:book} that then $R$ is at least ternary and one can
permute the coordinates in $R$ and then substitute 0 and 1 in $R$ in such a way that
that the ternary relation $R'(x,y,z)=R(x,y,z,0,\ldots,0,1,\ldots,1)$ has the
following properties:
\begin{enumerate}
\item $R'$ contains tuples $(1,1,1), (0,1,0), (1,0,0), (0,0,0)$, and

\item $R'$ does not contain the tuple $(1,1,0)$.
\end{enumerate}

Note that if $(0,0,1)\in R'$ then $R'(x,x,y)$ is $x\to y$. If $(0,0,1)\not\in R'$
then, since $R$ (and hence $R'$) is Horn (i.e., min-closed), at least one of of the
tuples $(1,0,1)$ and $(0,1,1)$ is not in $R'$. Then it is easy to check that at
least one of the relations $R'(x,y,x)$ and $R'(y,x,x)$ is $x\to y$. Hence, we can
use constraints of the form $x\to y$ when specifying instances of $\prob{\{R'\}}$.

We reduce \textsc{Minimum Dominating Set} to \prob{\{R'\}}. Let $G(V,E)$ be a graph
with $n$ vertices and $m$ edges where a dominating set of size at most $t$ has to be
found. Let $v_1$, $\dots$, $v_n$ be the vertices of $G$. Let $S=3m$. We construct a
formula with $nS+2m+1$ variables as follows:

\begin{itemize}
\item There is a special variable $x$.
\item For every $1 \le i \le n$, there are $S$ variables
  $x_{i,1}$, $\dots$, $x_{i,S}$. There is a constraint $x_{i,j}\to
  x_{i,j'}$ for every $1 \le j,j' \le n$.
\item For every $1 \le i \le n$, if $v_{s_1}$, $\dots$, $v_{s_d}$ are
  the neighbors of $v_i$, then there are $d$ variables $y_{i,1}$,
  $\dots$, $y_{i,d}$ and the following constraints:
$x_{s_1,1}\to y_{i,1}$,  $R'(x_{s_2,1},y_{i,1},y_{i,2})$,
  $R'(x_{s_3,1},y_{i,2},y_{i,3})$, $\dots$, $R'(x_{s_{d},1},y_{i,d-1},y_{i,d})$,
  $R'(x_{i,1},y_{i,d},x)$.
\item For every variable $z$, there is a constraint $x\to z$.
\end{itemize}
Observe that the number of variables of type $y_{i,j}$ is exactly $2m$. Setting
every variable to 1 is a satisfying assignment. Set $k:=St+S-1$.

Assume that there is a satisfying assignment where the number of 0's is at most $k$
(but positive). Variable $x$ has to be 0, otherwise every other variable is 1. If
$x_{i,1}$ is 0, then $x_{i,j}$ is 0 for every $1 \le j \le S$. Thus $k<S(t+1)$
implies that there are at most $t$ values of $i$ such that $x_{i,1}$ is 0. Let $D$
consist of all vertices $v_i$ such that $x_{i,1}$ is 0. We claim that $D$ is a
dominating set. Suppose that some vertex $v_i$ is not dominated. This means that if
$v_{s_1}$, $\dots$, $v_{s_d}$ are the neighbors of $v_i$, then the variables
$x_{s_1,1}$, $\dots$, $x_{s_d,1}$, $x_{i,1}$ all have the value 1. However, this
means that these variables force variables $y_{i,1}$, $\dots$, $y_{i,d}$ and
variable $x$ to value 1, a contradiction. Thus $D$ is a dominating set of size at
most $t$.

The reverse direction is also easy to see.  Assume that $G$ has a dominating set $D$
of size at most $t$. For every $1 \le i \le n$ and $1 \le j \le S$, set variable
$x_{i,j}$ to 1 if and only $v_i$ is not contained in $D$. Set $x$ to 0. It is easy
to see that this assignment can be extended to the variables $y_{i,j}$ to obtain a
satisfying assignment: indeed, if $v_{s_1}$, $\dots$, $v_{s_d}$ are the neighbors of
$v_i$ and none of them is in $D$ then $v_i\in D$, and we set $y_{i,1}=\ldots
=y_{i,d}=1$. Otherwise, if $j$ is minimal such that $v_{s_j}\in D$, we set
$y_{i,1}=\ldots =y_{i,j-1}=1$ and $y_{i,q}=0$ for $q\ge j$.
This satisfying assignment contains at most $St+2m+1 \le k$ variables with value 0,
as required.

Finally, we reduce $\prob{\{R'\}}$ to $\prob{\{R\}}$ (and so to
$\prob\Gamma$). Take an instance $(\phi,f,k)$ of $\prob{\{R'\}}$, let
$V$ be the variables of $\phi$ and $c_1,\ldots,c_p$ the constraints of
$\phi$. We build an instance $\phi'$ of
$\prob{\{R\}}$ as follows.
\begin{enumerate}
\item For each $1\le i \le \max(p,k+1)$, introduce new variables $x_0^i, x_1^i$.

\item For each constraint $c_i=R'(x,y,z)$ in formula $\phi$, replace
  it by  the constraint
$R(x,y,z,x_0^i,\ldots,x_0^i,x_1^i,\ldots,x_1^i)$.

\item For each ordered pair $(i,j)$ where $1\le i,j \le \max(p,k+1)$, add the
constraints $R(x_0^i,x_0^i,x_0^j,x_0^j,\ldots,x_0^j,x_1^j,\ldots,x_1^j)$ and
$R(x_1^j,x_1^j,x_1^i,x_0^j,\ldots,x_0^j,x_1^j,\ldots,x_1^j)$.
\end{enumerate}
Finally, extend $f$ so that, for all $i$, we have $x_0^i=0$ and $x_1^i=1$. It is
clear that the obtained mapping $f'$ is a solution to the new instance. Note that,
by the choice of $R'$, the tuple $(1,1,0,0,\ldots,0,1,\ldots,1)$ does not belong to
$R$. Hence, the constraints added in step (3) above ensure that if a variable of the
form $x_0^i$ or $x_1^i$ in $f'$ is flipped then, in order to get a solution to
$\phi'$ different from $f'$, one must flip at least one of $x_0^i$ or $x_1^i$ for
each $1\le i \le \max(p,k+1)$. Consequently, all solutions to $\phi'$ that lie within
distance $k$ from $f'$ must agree with $f'$ on all such variables. In other words,
searching for such a solution, it makes sense to flip only variables from $V$. Thus,
clearly, the instances $(\phi,f,k)$ and $(\phi',f',k)$ are equivalent.
 \qed\end{proof}

\bibliography{../../../bibtex/csp2}

\begin{thebibliography}{10}

\bibitem{Chapdelaine05:local}
P.~Chapdelaine and N.~Creignou.
\newblock The complexity of {B}oolean constraint satisfaction local search
  problems.
\newblock {\em Annals of Mathematics and Artificial Intelligence}, 43:51--63,
  2005.

\bibitem{Cohen06:handbook}
D.~Cohen and P.~Jeavons.
\newblock The complexity of constraint languages.
\newblock In F.~Rossi, P.~van Beek, and T.~Walsh, editors, {\em Handbook of
  Constraint Programming}, chapter~8. Elsevier, 2006.

\bibitem{Creignou01:book}
N.~Creignou, S.~Khanna, and M.~Sudan.
\newblock {\em Complexity Classifications of Boolean Constraint Satisfaction
  Problems}, volume~7 of {\em SIAM Monographs on Discrete Mathematics and
  Applications}.
\newblock 2001.

\bibitem{Crescenzi02:Hamming}
P.~Crescenzi and G.~Rossi.
\newblock On the {H}amming distance of constraint satisfaction problems.
\newblock {\em Theoretical Computer Science}, 288(1):85--100, 2002.

\bibitem{Dantsin02:local}
E.~Dantsin, A.~Goerdt, E.~Hirsch, R.~Kannan, J.~Kleinberg, C.~Papadimitriou,
  P.~Raghavan, and U.~Sch\"{o}ning.
\newblock A deterministic $(2-\frac{2}{k+1})^n$ algorithm for $k$-{SAT} based
  on local search.
\newblock {\em Theoretical Computer Science}, 289:69--83, 2002.

\bibitem{Downey99:book}
R.~Downey and M.~Fellows.
\newblock {\em Parameterized Complexity}.
\newblock Springer, 1999.

\bibitem{Fellows01:frontiers}
M.~Fellows.
\newblock Parameterized complexity: new developments and research frontiers.
\newblock In {\em Aspects of Complexity (Kaikura, 2000)}, volume~4 of {\em de
  Gruyter Series in Logic and Applications}, pages 51--72. 2001.

\bibitem{Flum06:book}
J.~Fl\"{u}m and M.~Grohe.
\newblock {\em Parameterized Complexity Theory}.
\newblock Springer, 2006.

\bibitem{Gu00:book}
J.~Gu, P.~Purdom, J.~Franko, and B.W. Wah.
\newblock {\em Algorithms for the Satisfiability Problem}.
\newblock Cambridge University Press, Cambridge, MA, 2000.

\bibitem{Hirsch00:local}
E.~Hirsch.
\newblock {SAT} local search algorithms: worst-case study.
\newblock {\em Journal of Automated Reasoning}, 24:127--143, 2000.

\bibitem{Hoos06:handbook}
H.~Hoos and E.~Tsang.
\newblock Local search methods.
\newblock In F.~Rossi, P.~van Beek, and T.~Walsh, editors, {\em Handbook of
  Constraint Programming}, chapter~5. Elsevier, 2006.

\bibitem{Kirousis03:minimal}
L.~Kirousis and Ph. Kolaitis.
\newblock The complexity of minimal satisfiability problems.
\newblock {\em Information and Computation}, 187:20--39, 2003.

\bibitem{Marx05:CSP}
D.~Marx.
\newblock Parameterized complexity of constraint satisfaction problems.
\newblock {\em Computational Complexity}, 14:153--183, 2005.

\bibitem{Marx08:TSP}
D.~Marx.
\newblock Searching the $k$-change neighborhood for {TSP} is {W}[1]-hard.
\newblock {\em Operations Research Letters}, 36:31--36, 2008.

\bibitem{Moore01:tiling}
C.~Moore and J.~M. Robson.
\newblock Hard tiling problems with simple tiles.
\newblock {\em Discrete Comput. Geom.}, 26(4):573--590, 2001.

\bibitem{Schaefer78:complexity}
T.J. Schaefer.
\newblock The complexity of satisfiability problems.
\newblock In {\em {STOC'78}}, pages 216--226, 1978.

\end{thebibliography}
\bibliographystyle{plain}

\clearpage

\appendix

\section*{Appendix}

The following proposition completes the proof of Theorem~\ref{thm:poly}.

\begin{proposition}\label{flip-not-w2aff}
If $R$ is a flip separable relation that is not width-2 affine then $\prob{\{R\}}$
is $\NP$-hard.
\end{proposition}

The above proposition will be proved through a sequence of lemmas,
the main lemmas being Lemma~\ref{affine-hard} and
Lemma~\ref{affine-flip-hard}.

We need three auxiliary lemmas that will be used in the subsequent proofs.

\begin{lemma}\label{add-rel}
Let $R$ is the set of solutions to a $\Gamma$-formula $\phi$, and let
$R'=\proj_J(R)$ where $J=\{j \mid \proj_j(R)=\{0,1\}\}$. Then the problems
$\prob{\Gamma}$ and $\prob{\Gamma\cup\{R'\}}$ are polynomial-time equivalent.
\end{lemma}

\begin{proof}
  Note that, for any fixed $j\not\in J$, each solution to $\phi$ takes
  the same value on the corresponding variable. In every instance of
  $\CSP{\Gamma\cup\{R'\}}$, every constraint of the form
  $c=R'(\ov{s})$ can be replaced by the constraints from $\phi$ where
  variables from $\ov{s}$ keep their places, while all other variables
  are new and do not appear elsewhere. This transformation (with the
  distance $k$ unchanged) is a polynomial-time reduction from
  $\prob{\Gamma\cup\{R'\}}$ to $\prob{\Gamma}$ \qed\end{proof}

Let $C_0=\{0\}$ and $C_1=\{1\}$, and let $\Gamma^c=\Gamma\cup\{C_0,C_1\}$.

\begin{lemma}\label{add-const}
If $\Gamma$ contains at least one of the relations $=$ or $\ne$ then $\prob{\Gamma}$
is polynomial-time equivalent to $\prob{\Gamma^c}$.
\end{lemma}

\begin{proof}
Assume that $\Gamma$ contains the equality relation, the other case
  is very similar.  We only need to provide a reduction from
  $\prob{\Gamma^c}$ to $\prob{\Gamma}$. Let
  $(\phi,f,k)$ be an instance of the former problem. For every
  variable $x$ in $\phi$ such that $\phi$ contains the constraint $C_0(x)$
  or $C_1(x)$ (i.e., the value of variable $x$ is
  forced to be a constant 0 or 1), remove this constraint, introduce $k$
  new variables $y_{x,1},\ldots,y_{x,k}$, and add new constraints
  $x=y_{x,i}$, ($i=1,\ldots,k$). Call the obtained $\Gamma$-formula
  $\phi'$.  Extend $f$ to a solution $f'$ of $\phi'$ by setting
  $f'(y_{x,i})=f(x)$ for all $i$, and consider the instance
  $(\phi',f',k)$ of $\prob{\Gamma}$. The equality constraints in
  $\phi'$ ensure that every solution to $\phi'$ whose distance is at
  most $k$ from $f'$ must coincide with $f'$ on all variables of the
  form $y_{x,i}$. It is clear now that the instances $(\phi,f,k)$ and
  $(\phi',f',k)$ are equivalent.
\qed\end{proof}

Let $R_1$ be an arbitrary $n$-ary relation ($n\ge 2)$ and let $R_2$ be
an $(n+1$)-ary relation such that
$(a_1 \zd a_n,a_{n+1})\in R_2 \Leftrightarrow (a_1
\zd a_n)\in R_1$ and $a_n\ne a_{n+1}$.

\begin{lemma}\label{add-neq}
The problems $\prob{\{R_1,\ne\}}$ and $\prob{\{R_2,\ne\}}$ are polynomial time
equivalent.
\end{lemma}

\begin{proof}
It is clear that $\prob{\{R_2,\ne\}}$ reduces to $\prob{\{R_1,\ne\}}$, by simply
replacing each constraint involving $R_2$ by its definition via $R_1$ and $\ne$. Let
us reduce $\prob{\{R_1,\ne\}}$ to $\prob{\{R_2,\ne\}}$. Let $(\phi,f,k)$ be an
instance of $\prob{\{R_1,\ne\}}$ where $\phi$ is over a set $V$ of variables. For
each variable $x\in V$, introduce two new variables $x',x''$ along with constraints
$x\ne x'$, $x'\ne x''$. Replace every constraint of the form $R_1(x_{i_1}\zd
x_{i_n})$ in $\phi$ by $R_2(x_{i_1}\zd x_{i_n},x'_{i_n})$, and leave all other
constraints in $\phi$ unchanged. Let $\phi'$ be the obtained instance of
$\CSP{\{R_2,\ne\}}$. Clearly, $f$ has a unique extension to a solution $f'$ to $I'$.
It is clear that the instance $(\phi',f',3k)$ of $\prob{\{R_2,\ne\}}$ is equivalent
to $(\phi,f,k)$.
\qed\end{proof}

We consider affine constraints first.

\begin{lemma}\label{lem:affine3}
If $\Gamma$ contains the ternary relation defined by equation $x+y+z=1$, then
$\prob{\Gamma}$ is \NP-hard.
\end{lemma}

\begin{proof}
First, we show that we can assume that the relation $\ne$ is in
$\Gamma$. Notice that if $R$ is the set of all solutions to the
instance $x+y+z=1, z+z+z=1$. Then $\proj_{\{3\}}(R)=1$ while
$\proj_{\{1,2\}}(R)$ is the equality relation. By
Lemma~\ref{add-rel}, we may assume that the relation $=$ is in
$\Gamma$. Then, by Lemma~\ref{add-const}, we can assume that
$C_0,C_1\in \Gamma$. By considering the instance $C_0(z),x+y+z=1$
and using Lemma~\ref{add-rel} again, we get that $\prob\Gamma$ is
polynomial-time equivalent to $\prob{\Gamma \cup\{\ne\}}$.

The hardness proof is by reduction from $\CSP{R_{\textsc{1-in-3}}}$, which is known
to be NP-hard even if every variable appears in exactly 3 constraints and each
  constraint contains 3 distinct variables \cite{Moore01:tiling}. (This implies that the
  number of variables equals the number of constraints and the weight
  of every solution is exactly $n/3$, where $n$ is the number of variables.) Given a
  $\CSP{R_{\textsc{1-in-3}}}$ formula
  $\phi$ with $n$ variables $x_1$, $\dots$, $x_n$, we construct a
  $\Gamma$-formula $\phi'$
  with variables
\begin{itemize}
\item $x_{i,j}$ for $1\le i \le n$, $1 \le j \le 2n$,
\item $v_t$ for every $1 \le t \le n$,
\item $y_j$ for every $0 \le j \le 2n^2$.
\end{itemize}
For every $1 \le i \le n$, $1 \le j,j' \le 2n$, we add the constraint
$x_{i,j}=x_{i,j'}$. For every $1 \le j,j' \le 2n^2$, we add the
constraint $y_j=y_{j'}$.
Assume that the $t$-th constraint in $\phi$ is on variables $x_a$, $x_b$, $x_c$. For
$1 \le j \le 2n$, we add the constraints $x_{a,j}+x_{b,j}+v_t=1$ and
$x_{c,j}+v_t+y_0=1$. Finally, we add the constraint
$y_0\neq y_1$.
 Define  assignment $f$ such that $f(v_t)=1$ for $1\le t \le n$
and $f(y_j)=1$ for $1\le i \le 2n^2$, and every other variable is 0. The weight of
$f$ is $2n^2+n$. Set $k:=2n^2+1+n+2n\cdot n/3$. This completes the description of
the reduction.

Assume that $\phi$ has a solution $f_0$. Define $f'$ such that
$f'(x_{i,j})=f_0(x_i)$ ($1 \le i \le n$, $1 \le j \le 2n$), $f'(y_0)=1$, $f'(y_j)=0$
($1 \le j \le 2n^2$). This assignment can be extended to each $v_t$ in a unique way:
if the $t$-th constraint in $\phi$ is on variables $x_a$, $x_b$, $x_c$, then exactly
two of the variables $x_{a,1},x_{b,1},x_{c,1},y_0$ have  value $1$, hence we can set
$v_t$ accordingly. Thus we can obtain a satisfying assignment $f'$ this way. Observe
that the weight of $f_0$ is exactly $n/3$: each variable with value 1 in $f_0$
appears in exactly 3 constraints and each constraint contains exactly one such
variable. Thus the weight of $f'$ is at most $2n\cdot n/3+1+n<2n^2+n$, strictly less
than the weight of $f$. Assignments $f$ and $f'$ differ on at most $2n^2+1+n+2n\cdot
n/3=k$ variables.

Assume now that $\phi'$ has a satisfying assignment $f'$ with $w(f')<w(f)$ and
$\dist(f,f')\le k$.
We claim that
$f'(y_0)=1$ and $f'(y_i)=0$ ($1 \le i \le 2n^2$). Otherwise, $w(f')<w(f)$ would
imply that $f'(v_t)=0$ for at least one $1\le t \le n$. But this means that there is
at least one $1 \le i \le n$ such that $f'(x_{i,j})=1$ for every $1\le j \le 2n$.
Thus the weight of $f'$ is at least $2n^2+2n>2n^2+n$, a contradiction.

Define $f_0(x_i):=f'(x_{i,1})$ for every $1\le i \le n$. The weight of
$f_0$ is at most $n/3$: otherwise, $\dist(f,f')$ would be at least
$2n(n/3+1)+2n^2+1>k$. Let $x_a$, $x_b$, $x_c$ be the variables in the
$t$-th constraint in $\phi$. The facts
$f'(x_{a,1})+f'(x_{b,1})+f'(v_t)=1$, $f'(x_{c,1})+f'(v_t)+f'(y_0)=1$
and $f'(y_0)=1$ imply that at least one of $f'(x_{a,1})$,
$f'(x_{b,1})$, $f'(x_{c,1})$ is 1. Thus if we denote by $X$ the set of
those variables of $\phi$ that have value $1$ in $f_0$, then each
constraint in $\phi$ contains at least one variable of $X$.  Moreover,
it is not possible that a constraint contains more than one variables
of $X$. To see this, observe that each variable is contained in 3
constraints, thus $|X|\le n/3$ implies that the variables of $X$
appear in at most $n$ constraints. However, we have seen that each
constraint contains at least one variable of $X$, hence the variables
of $X$ appear in {\em exactly} $n$ contraints. Equality is possible only
if any two variables of $X$ appear in disjoint constraints, that is,
no constraint contains more than one variables from $X$.  Therefore,
$f_0$ is a solution for the instance $\phi$. \qed\end{proof}

Define the 6-ary affine relation $R_6$ by the following system of equations:
$x_1+x_2+x_3=1, x_i+x_{i+3}=1 (1\le i\le 3)$. Let $J$ be a subset of $\{1\zd 6\}$,
and let $R_J$ denote the projection of $R_6$ onto the set of indices $J$.

\begin{lemma}\label{base-aff-rel}
$\prob{\{R_J\}}$ is \NP-hard for every $J\subseteq \{1,\dots,6\}$ containing at
least one number from each pair $i,i+3$ where $i=1,2,3$.
\end{lemma}

\begin{proof}
Note that the relation $R_6$ does not change if the first equation in the system is
replaced by $x_4+x_2+x_3=0$. We will assume that $1,2,3\in J$. All other cases are
very similar. As in the (beginning of the) proof of Lemma~\ref{lem:affine3}, we can
argue that $\prob{\{R_J\}}$ is polynomial-time equivalent to $\prob{\{R_J,\ne\}}$
Now the lemma follows from Lemmas~\ref{lem:affine3} and~\ref{add-neq}.
\qed\end{proof}

\begin{lemma}\label{affine-hard}
$\prob{\{R\}}$ is \NP-hard if $R$ is affine, but not of width 2.
\end{lemma}
\begin{proof}
  First, we show that it is enough to prove the result for
  $\Gamma=\{R,C_0,C_1\}$, that is, that \prob{\Gamma} and \prob{\{R\}} are
  polynomial-time equivalent. If $R$ is neither 0- nor 1-valid, then
  fix a tuple in $R$ that contains both 0 and 1, say $(0\zd 0,1\zd
  1)$, and consider the relation $R'(x,y)=R(x\zd x,y\zd y)$, where
  positions of the $x$'s correspond to the 0's in the tuple. It is clear
  that either $R_1=\{(0,1)\}$ in which case we are done, or $R_1$ is
  the disequality relation, and then we are done by
  Lemma~\ref{add-const}. If $R$ is 0-valid, but not 1-valid, then take
  a tuple $\ov{a}\in R$ that contains a minimal positive number of
  1's, say $(0\zd 0,1\zd 1)$. The fact that $R$ is affine and 0-valid
  implies that it can be defined by a system of linear equations of
  the form $\sum{x_i}=0$. It follows that the number of 1's in
  $\ov{a}$ is at least two. Consider the instance $R(z\zd z,x,y\zd y),R(z\zd z)$
  where the $z$'s in the first constraint correspond to the 0's in $\ov{a}$.
  by the choice of $\ov{a}$, the projection on $x,y$ of the set of solutions of this
  instance is the equality relation, while the projection on $z$ is $C_0$.
  Now we can use Lemmas~\ref{add-rel} and~\ref{add-const}.
  The argument is very similar if $R$ is 1-valid, but not 0-valid. If $R$ is both 0- and
  1-valid then take any tuple, say $(0\zd 0,1\zd 1)\not\in R$ and
  consider the relation $R'(x,y)=R(x\zd x,y\zd y)$, where positions of
  the $x$'s correspond to 0's in the tuple. Since $R'$ is affine, it
  can only be the equality relation, and we are done again.

Now let $R'$ be a minimium arity relation which can be obtained from $R$ by
substituting constants and identifying variables (note that such a relation is also
affine) and which is not of width-2. We claim that $R'$, possibly after permuting
variables, is a relation of the form $R_J$, as in Lemma~\ref{base-aff-rel}. By
Lemma~\ref{add-rel}, establishing this claim would finish the proof of the present
lemma.

Note that by the minimality of (the arity of) $R'$, none of the
projections of $R'$ onto a single coordinate can be one-element, and
none of the projections of $R'$ onto a pair of coordinates can be the
equality relation.  Thus every binary projection is either the
disequality relation or $\{0,1\}^2$ (note that the binary projection
cannot contain 3 tuples, since $R'$ is affine).  Furthermore, if two different binary
projections of $R'$ are disequality relations then the corresponding
pairs of coordinates are disjoint. Let $R''$ be a
largest arity projection of $R'$ such that every binary projection of
$R''$ is $\{0,1\}^2$. Note that, to prove the lemma, it is sufficient
to show that $R''$ is ternary and can be described by a single
equation $x_1+x_2+x_3=a$, $a\in\{0,1\}$. It is easy to see that then a
relation of the form $R_J$ can obtained from $R'$ by permuting
coordinates: every variable of $R$ not in $\{x_1,x_2,x_3\}$ forms a
disequality relation with one of $x_1,x_2,x_3$.

Consider the relations $R''_0=R''(x_1\zd x_n,0)$ and $R''_1=R''(x_1\zd x_n,1)$. By
the choice of $R'$, both of these relations are of width-2 affine, that is each of
them can be expressed by a system of equations of the form $x_i+x_j=a$,
$a\in\{0,1\}$. Note that it is impossible that, for a pair of coordinates $i,j$,
exactly one of the projections of $R''_0$ and $R''_1$ onto $i,j$ is the full
relation $\{0,1\}^2$. Indeed, this would imply that the size of the projection of
$R''$ onto $\{i,j\}$ would not be a power of 2. Therefore, if, for a pair of indices
$i,j$, one of the projections of $R''_0$ and $R''_1$ onto $i,j$ is the equality
relation then the other must be the disequality relation (to ensure that the
projection of $R''$ onto $\{i,j\}$ is $\{0,1\}^2$). It follows that $R''$ can be
described by the following system of equations: for each pair $i,j$ such that
$\proj_{i,j}(R''_0)$ is the equality relation the system contains the equation
$x_i+x_j+x_{n+1}=0$, and for each pair $i,j$ such that $\proj_{i,j}(R''_0)$ is the
disequality relation the system contains the equation $x_i+x_j+x_{n+1}=1$, and there
are no other equations in the system.

Note that if some $x_i$ with $1\le i\le n$ participates in at least two equations in
this system (which also involve $x_j$ and $x_{j'}$) then the projection of $R''$
onto $j,j'$ would not be $\{0,1\}^2$, which is a contradiction. Hence, the only
variable that may appear in more than one equation is $x_{n+1}$. Assume that the
system contains at least two  equations, say, $x_1+x_2+x_{n+1}=a$ and
$x_3+x_4+x_{n+1}=b$. Then,  by identifying $x_4$ with $x_{n+1}$ in $R'$, we would be
able to obtain a relation which is affine but not width-2 (because of the equation
$x_1+x_2+x_{n+1}=a$), and has arity smaller than the arity of $R'$, which is a
contradiction.

The lemma is proved.
\qed\end{proof}

It remains to consider the case when $\Gamma$ contains a relation that is flip
separable, but not affine. Again, we first consider one particular relation.

\begin{lemma}\label{lem:1in3}
$\prob{\{R_{\textsc{1-in-3}}\}}$ is \NP-hard.
\end{lemma}

\begin{proof}
First, we show that we can assume that the relation $\ne$ is available.
Consider the instance
$R_{\textsc{1-in-3}}(x,y,z_0),R_{\textsc{1-in-3}}(z_0,z_0,z_1)$.
It is easy to see that the projection on $x,y$ of the set of all
solutions of this instance is $\neq$, and Lemma~\ref{add-rel}
implies that it is now sufficient to prove that
$\prob{\{R_{\textsc{1-in-3}},\neq \}}$ is \NP-hard.

The proof is by reduction from \prob{\{x\vee y\}}. Let $(\phi_1,f_1,k)$ be an
instance of \prob{\{x\vee y\}} with $n$ variables and $m$ constraints. Set
$S:=10n^2m^2$. We construct a formula $\phi_2$ with the following variables:
\begin{itemize}
\item For each variable $v$ of $\phi_1$ with $f(v)=1$, there is a variable $x^0_v$ and $S-2m$ variables
  $x^1_{v,1}$, $\dots$, $x^1_{v,S-2m}$.
\item For each variable $v$ of $\phi_1$ with $f(v)=0$, there is a variable $x^0_v$ and $S$ variables
  $x^1_{v,1}$, $\dots$, $x^1_{v,S}$.
\item For each constraint $c$ of $\phi_1$, there is a variable $y_c$.
\end{itemize}
The constraints of $\phi_2$ are as follows:
\begin{itemize}
\item For every variable $v$ of $\phi_1$ with $f(v)=1$, there is a constraint
  $x^0_v\neq x^1_{v,i}$ ($1 \le i \le S-2m$).
\item For every variable $v$ of $\phi_1$ with $f(v)=0$, there is a constraint
  $x^0_v\neq x^1_{v,i}$ ($1 \le i \le S$).
\item For every constraint $c=(u\vee v)$ of $\phi_1$, there is a
  constraint $R_{\textsc{1-in-3}}(x^0_u,x^0_v,y_c)$.
\end{itemize}

We observe that the satisfying assignments of $\phi_2$ are in one-to-one
correspondence with the satisfying assignments of $\phi_1$. First, since $x^0_v\neq
x^1_{v,i}$ for every $i$, we have that $x^1_{v,i_1}=x^1_{v,i_2}$ for every
$i_1,i_2$. Given a satisfying assignment $f$ of $\phi_2$, we can obtain an
assignment of $\phi_1$ by setting $f(v)=f(x^1_{v,1})$. Observe that this gives a
satisfying assignment of $\phi_1$: if $c=(u \vee v)$ is a constraint of $\phi_1$,
then $f(u)=f(v)=0$ would imply $f(x^0_{u})=f(x^0_{v})=1$, which would mean that the
constraint $R_{\textsc{1-in-3}}(x^0_u,x^0_v,y_c)$ is not satisfied.
 On the other hand, if $f$ satisfies $\phi_1$, then setting
 $f(x^1_{v,1})=f(v)$ has a unique satisfying extension to 
 $\phi_2$.

Let $f_2$ be the assignment of $\phi_2$ corresponding to assignment $f_1$ of
$\phi_1$. Set $k':=k(S+1)+m$.  We claim that $(\phi_2,f_2,k')$ is a yes-instance if
and only if $(\phi_1,f_1,k)$ is a yes-instance.

Let $f'_1$ be an arbitrary satisfying
 assignment of $\phi_1$ and let $f'_2$ be the corresponding assignment
 of $\phi_2$. Let $B$ (resp., $B'$) contain variable $v$ of
 $\phi_1$ if and only if $f_1(v)=1$ (resp., $f'_1(v)=1$).  Assignments
 $f_2$ and $f'_2$ differ on every variable $x^0_v$, $x^1_{v,i}$ with
 $v\in |B\setminus B'|\cup |B'\setminus B|$, and possibly on some of
 the variables $y_c$.  Counting the number of 0's and 1's on these
 variables, it follows that the weight of $f'_2$ is at least
\begin{equation}\label{eq:low}
w(f_2)+|B'\setminus B|(S-1)-|B\setminus B'|(S-2m-1)-m
\end{equation}
and at most
\begin{equation}\label{eq:up}
w(f_2)+|B'\setminus B|(S-1)-|B\setminus B'|(S-2m-1)+m,
\end{equation}
where the $\pm m$ term accounts for the variables of the form $y_c$.

To prove that $(\phi_2,f_2,k')$ is a yes-instance if and only if
$(\phi_1,f_1,k)$ is a yes-instance, first we assume that
$(\phi_1,f_1,k)$ has a solution $f'_1$; let $f'_2$ be the
corresponding assignment of $\phi_2$. Define $B$ and $B'$ as above.
It is easy to see that the distance of $f_2$ and $f'_2$ is at most
$|B\setminus B'|(S-2m+1)+|B'\setminus B|(S+1)+m\le k(S+1)+m=k'$. From
$|B'|<|B|$ we have $|B'\setminus B|<|B\setminus B'|$. Thus by
\eqref{eq:up}, the weight of $f'_2$ is at most
\begin{multline*}
w(f_2)+|B'\setminus B|(S-1)-|B\setminus B'|(S-2m-1)+m\\\le w(f_2)+|B'\setminus
B|(S-1)-(|B'\setminus B|+1)(S-2m-1)+m\\
= w(f_2)+|B'\setminus B|2m+3m+1-S <w(f_2),
\end{multline*}
since $S$ is sufficiently large.

For the second direction, suppose that there is a solution $f'_2$ for instance
$(\phi_2,f_2,k')$. Let $f'_1$ be the corresponding satisfying assignment, and define
$B$ and $B'$ as before. First we show that the distance of $f'_1$ and $f_1$ is at
most $k$, i.e., $|B'\setminus B|+|B\setminus B'|\le k$. Assignments $f_2$ and $f'_2$
differ on every vertex of $x^1_{v,i}$ for $v\in |B'\setminus B|+|B\setminus B'|$ and
$1 \le i \le S-2m$. Thus $|B'\setminus B|+|B\setminus B'|>k$ would imply that the
distance of $f_2$ and $f'_2$ is at least $(k+1)(S-2m)= k(S+1)+S-k-2m> k(S+1)+m=k'$,
since $S$ is sufficiently large. Let us show that $|B'\setminus B|<|B\setminus B'|$.
Suppose first that $|B\setminus B'|=0$. By \eqref{eq:low}, the weight of $f'_2$ is
at least
\[
w(f_2)+|B'\setminus B|(S-1)-|B\setminus B'|(S-2m-1)-m\ge w(f_2)+(S-1)-m\ge w(f_2),
\]
contradicting the assumption $w(f'_2)<w(f_2)$. Suppose now that $|B'\setminus B|\ge
|B\setminus B'| \ge 1$. Now the weight of $f'_2$ is at least
\begin{multline*}
  w(f_2)+|B'\setminus B|(S-1)-|B\setminus B'|(S-2m-1)-m \\\ge w(f_2)+ |B\setminus
B'|(S-1)-|B\setminus B'|(S-2m-1)-m = w(f_2) + |B\setminus B'|2m-m >w(f_2),
\end{multline*}
again a contradiction.
\qed\end{proof}

Consider the 6-ary relation $T_6$ defined as follows: $(a_1\zd
a_6)\in T_6$ if and only if $(a_1,a_2,a_3)\in R_{\textsc{1-in-3}}$
and $a_i\ne a_{i+3} (1\le i\le 3)$. Let $J$ be a subset of $\{1\zd
6\}$ and let $T_J$ denote the projection of $T_6$ onto the set of
indices $J$. Similarly to Lemma~\ref{base-aff-rel}, we have:

\begin{lemma}\label{base-flip-rel}
$\prob{\{T_J\}}$ is \NP-hard if $J$ contains at least one number
from each pair $i,i+3$ where $i=1,2,3$.
\end{lemma}

\begin{lemma}\label{affine-flip-hard}
$\prob{\{R\}}$ is
\NP-hard if $R$ is flip separable but not affine.
\end{lemma}
\begin{proof}
Assume first that $R$ is 0-valid. Then we can assume that the
relation $C_0$ is available. Indeed, if $R$ is not 1-valid then
$R(x\zd x)$ is this relation. If $R$ is also 1-valid then we can
get $C_0$ and $C_1$ as in the proof of Lemma~\ref{affine-hard}.
Note that a relation obtained from a flip separable relation by
substituting constants and identifying coordinates is again flip
separable. It is shown in the proof of Lemma 5.30
of~\cite{Creignou01:book} that, since $R$ is not affine, it is
possible to permute coordinates in $R$ in such a way that the
relation $R'(x,y,z)$ defined as $R(0\zd 0,x\zd x,y\zd y,z\zd z)$
satisfies the following properties:
\begin{itemize}
\item The tuples $(0,0,0),(0,1,1),(1,0,1)$ belong to $R'$, and
\item the tuple $(1,1,0)$ does not belong to $R'$.
\end{itemize}
Since $R'$ is flip separable, it is easy to see that if one of the tuples from
$R_{\textsc{1-in-3}}$ belongs to $R'$ then all of them do. However, in
this case $(1,0,0)\in R'$ has flip
sets $\{1\}$ and $\{1,2\}$, and therefore $\{2\}$ should also be  a flip set
for this tuple, which is impossible because $(1,1,0)\not\in R'$. It follows
that $R'\cap R_{\textsc{1-in-3}}=\emptyset$. If the tuple $(1,1,1)$ was in
$R'$ then it would have flip sets $\{1,2,3\}$ and $\{1\}$, and hence $\{2,3\}$
as well, which is impossible. It follows that $R'=\{(0,0,0),(0,1,1),(1,0,1)\}$.
It is easy to see that $R'$ is obtained from $T_{\{2,3,4\}}$ (as in
Lemma~\ref{base-flip-rel}) by a cyclic permutation of coordinates.
The lemma  now follows from Lemmas~\ref{add-rel} and~\ref{base-flip-rel}.

If $R$ is 1-valid then a similar reasoning shows that one can obtain the
relation $\{(1,1,1),(0,0,1),(0,1,0)\}$ which is the relation $T_{\{3,4,5\}}$
with permuted coordinates.

Assume now that $R$ is neither 0- nor 1-valid. As in the proof of
Lemma~\ref{affine-hard}, one can show that both $C_0$ and $C_1$ can be assumed
to be available. It is shown in the proof of Lemma 5.30
of~\cite{Creignou01:book} that if $R$ is Horn then there is an instance of
$\CSP{\{R,C_0,C_1\}}$ over variables $\{x,y,z,u_0,u_1\}$ which contains
constraints $C_0(u_0)$ and $C_1(u_1)$, and such that the projection on $x,y,z$
of the set of all solutions to the instance contains
$(0,0,0),(0,1,1),(1,0,1)$, but does not contain $(1,1,0)$. This set must be a
Horn relation. On the other hand, we showed above that this set must coincide
with $\{(0,0,0),(0,1,1),(1,0,1)\}$, which is not Horn, a contradiction. If $R$
is not Horn, then it is easy to show (see the proof of the Lemma~5.26
of~\cite{Creignou01:book}) that there is an instance of
$\CSP{\{R,C_0,C_1\}}$ over variables $\{x,y,u_0,u_1\}$
which contains constraints $C_0(u_0)$ and $C_1(u_1)$, and
such that the projection on $x,y$ is not $\{0,1\}^2$ and
contains the relation $\ne$. Since this set must be a flip separable relation,
it is equal to $\ne$. Again, it is shown in the proof of Lemma 5.30
of~\cite{Creignou01:book} that there exists an instance of
$\CSP{\{R,\ne,C_0,C_1\}}$ over variables $\{x,y,z,x',y',z',u_0,u_1\}$ which
contains constraints $x\ne x'$, $y\ne y'$, $z\ne z'$, $C_0(u_0)$ and
$C_1(u_1)$, and such that the projection on $x,y,z$ of the set $R''$ of all
solutions to the instance contains $(0,0,0),(0,1,1),(1,0,1)$, but does not
contain $(1,1,0)$.
Since $R''$ is flip separable, it is easy to show (as above) that, in fact,
the projection of $R''$ onto $x,y,z,x',y',z'$ is equal to $T_6$ (with permuted
coordiantes). Now the result follows from by Lemmas~\ref{add-rel}
and~\ref{base-flip-rel}.
\qed\end{proof}

\end{document}